\documentclass[11pt,reqno]{amsart}
\usepackage{fullpage}
\usepackage{a4wide}
\usepackage[utf8x]{inputenc}
\usepackage[T1]{fontenc}
\usepackage[english]{babel}
\usepackage{hyperref}
\usepackage{amsfonts,amsmath,amssymb,amsthm,amsrefs}
\usepackage{latexsym}
\usepackage{layout}
\usepackage{dsfont}
\usepackage{tikz}
\usepackage{parskip}

\newtheorem{theorem}{Theorem}

\newtheorem{lemma}[theorem]{Lemma}
\newtheorem{corol}[theorem]{Corollary}
\theoremstyle{definition}

\theoremstyle{remark}
\newtheorem{remark}[theorem]{Remark}

\newcommand{\me}{\mathrm{e}}

\newcommand{\drift}{c}

\def\beq{\begin{eqnarray}}
\def\eeq{\end{eqnarray}}

\def\al*#1{\begin{align*}#1\end{align*}}

\def\ga*#1{\begin{gather*}#1\end{gather*}}

\def\alat*#1#2{\begin{alignat*}{#1}#2\end{alignat*}}
\def\bea{\begin{eqnarray*}}
\def\eea{\end{eqnarray*}}
\def\ml*#1{\begin{multline*}#1\end{multline*}}

\allowdisplaybreaks

\begin{document}
\title[]{On occupation times in the red of L\'{e}vy risk models}
\author[]{David Landriault}
\address{Department of Statistics and Actuarial Science, University of
Waterloo, Waterloo, ON, N2L 3G1, Canada}
\email{dlandria@uwaterloo.ca}
\author[]{Bin Li}
\address{Department of Statistics and Actuarial Science, University of
Waterloo, Waterloo, ON, N2L 3G1, Canada}
\email{bin.li@uwaterloo.ca}
\author[]{Mohamed Amine Lkabous}
\address{Department of Statistics and Actuarial Science, University of
Waterloo, Waterloo, ON, N2L 3G1, Canada}
\email{mohamed.amine.lkabous@uwaterloo.ca}
\date{\today }

\begin{abstract}
In this paper, we obtain analytical expression for the distribution of the
occupation time in the red (below level $0$) up to an (independent)
exponential horizon for spectrally negative L\'{e}vy risk processes and
refracted spectrally negative L\'{e}vy risk processes. This result improves
the existing literature in which only the Laplace transforms are known. Due
to the close connection between occupation time and many other quantities,
we provide a few applications of our results including future drawdown,
inverse occupation time, Parisian ruin with exponential delay, and the last
time at running maximum. By a further Laplace inversion to our results, we
obtain the distribution of the occupation time up to a finite time horizon
for refracted Brownian motion risk process and refracted Cram\'{e}r-Lundberg
risk model with exponential claims.
\end{abstract}

\maketitle


\section{Introduction}

Occupation times measure the amount of time a stochastic process stays in a
certain region. It is a long-standing research topic in applied probability
and has wide applications in many fields. In finance,
occupation-time-related financial derivatives were studied under various
dynamics for the underlying asset (e.g., Cai et al. \cite{caietal2010}, Cai
and Kou \cite{caietal2011} and Linetsky \cite{linetsky1999}). In actuarial
mathematics, occupation times can naturally be used as a measure of the risk
inherent to an insurance portfolio. The occupation time of an insurer's
surplus process below a given threshold level (often chosen to be the
\textquotedblleft symbolic\textquotedblright\ level $0$) is particularly
critical in the assessment of an insurer's solvency risk (e.g., Landriault
et al. \cite{landriaultetal2014} and Gu\'{e}rin and Renaud \cite%
{guerinrenaud2015}). With this application in mind, we define the occupation
time in the red of a risk process $X$ in the time interval $(0,t)$ as 
\begin{equation*}
\mathcal{O}^{X}_{t}=\int_{0}^{t}\mathbf{1}_{\left( -\infty ,0\right) }\left(
X_{s}\right) \mathrm{d}s,  \label{Occupation}
\end{equation*}%
where 
\begin{equation*}
\mathbf{1}_{A}\left( x\right) =\left\{ 
\begin{array}{ll}
1,\text{ }x\in A, &  \\ 
0,\text{ otherwise.} & 
\end{array}%
\right.
\end{equation*}%
Its infinite-time counterpart $\mathcal{O}^{X}_{\infty }$ is defined as $%
\mathcal{O}^{X}_{\infty }=\int_{0}^{\infty }\mathbf{1}_{\left( -\infty
,0\right) }\left( X_{s}\right) \mathrm{d}s$.

There exists a number of results on the occupation time $\mathcal{O}^{X}_{t}$
in the literature. For the standard Brownian motion, the distribution of $%
\mathcal{O}^{X}_{t}$ appeared in L\'{e}vy's \cite{levy1940} famous arc-sine
law. This formula was generalized by Akahori \cite{akahori1995} and Tak\'{a}%
cs \cite{takacs1996} to a Brownian motion with drift. For the classical
compound Poisson process with some special jump distributions, Dos Reis \cite%
{dosreis1993} obtained the moment generating function of $\mathcal{O}%
^{X}_{\infty }$ using a martingale approach. Zhang and Wu \cite{zhangwu2002}
further solved the Laplace transform of $\mathcal{O}^{X}_{\infty }$ by
considering a compound Poisson process perturbed by an independent Brownian
motion.

Spectrally negative L\'{e}vy process is a wide used more general risk model,
which includes Brownian motions (with drift) and compound Poisson processes
(with diffusion) as special cases. Under this framework, Landriault et al. 
\cite{landriaultetal2011} first derived the Laplace transform of $\mathcal{O}%
^{X}_{\infty }$ in terms of the process' scale function. Loeffen et al. \cite%
{loeffenetal2014} generalized the results in \cite{landriaultetal2011} by
characterizing the joint Laplace transform of $\left( \tau _{b}^{+},\mathcal{%
O}^{X}_{\tau _{b}^{+}}\right) $ where $\tau _{b}^{+}=\inf \{t>0\colon
X_{t}>b\}$. Li and Palmowski \cite{li-palmowski2017} considered a further
extension by studying weighted occupation times. Also, potential measures
involving occupation times were also analyzed by Gu\'{e}rin and Renaud \cite%
{guerin_renaud_2015} and Li et al. \cite{lietal2015}.

Admittedly, it is most desirable to solve the distribution of $\mathcal{O}^{X}_{t}$, the occupation time up to a finite time horizon, but it is a highly
challenging problem. To the best of our knowledge, the distribution of $%
\mathcal{O}^{X}_{t}$ has been found for only two types of processes:
Brownian motions with drift by Akahori \cite{akahori1995} and compound
Poisson with exponential jumps by Gu\'{e}rin and Renaud \cite%
{guerinrenaud2015}.

The main theoretical result of this paper is that, for spectrally negative L%
\'{e}vy risk processes and refracted spectrally negative L\'{e}vy risk
processes, we obtain an analytical expression for the distribution of $%
\mathcal{O}^{X}_{ \mathrm{e}_{\lambda }}$, where $\mathrm{e}_{\lambda }$
denotes an independent exponential time horizon with rate $\lambda>0$. Note
that refracted spectrally negative L\'{e}vy risk processes was introduced by
Kyprianou and Loeffen \cite{kyprianouloeffen2010}). This class of processes
is of interest in a number of insurance applications. This includes dividend
payouts under a threshold strategy (e.g., Hern\'{a}ndez-Hern\'{a}ndez \cite%
{HH2016} and Czarna et al. \cite{czarna2018}) and variable annuities with a
state-dependent fee structure (e.g., Bernard et al. \cite%
{bernard_hardy_mackay_2014}). Results on occupation times on the refracted
spectrally negative L\'{e}vy process can be found in Kyprianou et al. \cite%
{kyprianou2014occupation}, Renaud \cite{renaud2014} and Li and Zhou \cite%
{li-zhou2018}. Note that in these papers, the aim was to identify the
Laplace transform of some occupation times while in this paper we partially
generalize their results by solving the distribution.

The main contributions of our paper are summarized below. First, since the
Laplace transform of $\mathcal{O}^{X}_{\mathrm{e}_{\lambda }}$, namely $%
\mathbb{E}\left[ \mathrm{e}^{-q\mathcal{O}^{X}_{\mathrm{e}_{\lambda }}}%
\right] $, is known in the literature (see Lemma 1 below), one may obtain
numerically the distribution of $\mathcal{O}^{X}_{t} $ via a double Laplace transform inversion (with respect to $q$ and $\lambda $).
However, it is well known that numerical Laplace inversion methods suffer from various stability issues, that may lead to significant computational
errors. By deriving a general expression for the distribution of $\mathcal{O}^{X}_{ \mathrm{e}_{\lambda }}$, we explicitly invert one Laplace transform, providing further structure to the problem in addition to saving one round of (numerical) Laplace inversion. Second,
since occupation time is closely related to many other quantities, we
provide a few applications of our results in Section 2.3 including future
drawdown, inverse occupation time, Parisian ruin with exponential delays, and
the last time at running maximum. Third, in addition to aforementioned
Brownian motions with drift and compound Poisson with exponential jumps, we
obtain the distribution of $\mathcal{O}^{X}_{t}$ for two more models in Section \ref{section6}: refracted Brownian motion risk process and refracted Cram\'{e}r-Lundberg
risk model with exponential claims.

The rest of the paper is organized as follows. In Section \ref{section1}, we
first present the necessary background material on spectrally negative L\'{e}%
vy processes and scale functions. We then derive the main results of this
paper and consider some relevant applications. We conclude this section by
providing some examples of L\'{e}vy risk processes. In Section \ref{section6}%
, we extend our study in parallel to the class of refracted spectrally
negative L\'{e}vy process.

\section{Occupation times of spectrally negative L\'{e}vy processes}

\label{section1}

\subsection{Preliminaries}

First, we present the necessary background material on spectrally negative L\'{e}vy processes. A L\'{e}vy insurance risk process $X$ is a process
with stationary and independent increments and no positive jumps. To avoid
trivialities, we exclude the case where $X$ has monotone paths. As the L\'{e}%
vy process $X$ has no positive jumps, its Laplace transform exists: for all $%
\lambda, t \geq 0$, 
\begin{equation*}
\mathbb{E} \left[ \mathrm{e}^{\lambda X_t} \right] = \mathrm{e}^{t
\psi(\lambda)} ,
\end{equation*}
where 
\begin{equation*}
\psi(\lambda) = \gamma \lambda + \frac{1}{2} \sigma^2 \lambda^2 +
\int^{\infty}_0 \left( \mathrm{e}^{-\lambda z} - 1 + \lambda z \mathbf{1}%
_{(0,1]}(z) \right) \Pi(\mathrm{d}z) ,
\end{equation*}
for $\gamma \in \mathbb{R}$ and $\sigma \geq 0$, and where $\Pi$ is a $%
\sigma $-finite measure on $(0,\infty)$ called the L\'{e}vy measure of $X$
such that 
\begin{equation*}
\int^{\infty}_0 (1 \wedge z^2) \Pi(\mathrm{d}z) < \infty .
\end{equation*}
Throughout, we will use the standard Markovian notation: the law of $X$ when
starting from $X_0 = x$ is denoted by $\mathbb{P}_x$ and the corresponding
expectation by $\mathbb{E}_x$. We write $\mathbb{P}$ and $\mathbb{E}$ when $%
x=0$.\newline
We recall the definitions of standard first-passage stopping times : for $b
\in \mathbb{R}$, 
\begin{align*}
\tau_b^- &= \inf\{t>0 \colon X_t<b\} \quad \text{and} \quad \tau_b^+ =
\inf\{t>0 \colon X_t > b\} ,
\end{align*}
with the convention $\inf \emptyset=\infty$.

We now present the definition of the scale functions $W_{q}$ and $Z_{q}$ of $%
X$. First, recall that there exists a function $\Phi \colon [0,\infty) \to
[0,\infty)$ defined by $\Phi_{q} = \sup \{ \lambda \geq 0 \mid \psi(\lambda)
= q\}$ (the right-inverse of $\psi$) such that 
\begin{equation*}
\psi ( \Phi_{q} ) = q, \quad q \geq 0 .
\end{equation*}
When $\mathbb{E}[X_1]>0$, we have 
\begin{equation}  \label{limesp}
\lim_{q \rightarrow 0}\dfrac{q}{\Phi_q}=\psi^{\prime }(0+)=\mathbb{E}[X_1].
\end{equation}
Now, for $q \geq 0$, the $q$-scale function of the process $X$ is defined as
the continuous function on $[0,\infty)$ with Laplace transform 
\begin{equation}  \label{def_scale}
\int_0^{\infty} \mathrm{e}^{- \lambda y} W_{q} (y) \mathrm{d}y = \frac{1}{%
\psi_q(\lambda)} , \quad \text{for $\lambda > \Phi_q$,}
\end{equation}
where $\psi_q(\lambda)=\psi(\lambda) - q$. This function is unique, positive
and strictly increasing for $x\geq0$ and is further continuous for $q\geq0$.
We extend $W_{q}$ to the whole real line by setting $W_{q}(x)=0$ for $x<0$.
We write $W = W_{0}$ when $q=0$. The initial values of $W_{q}$ are known to be
\begin{equation*}
\begin{split}
W_{q}(0+) &=
\begin{cases}
1/\drift & \text{when X has bounded variation;,} \\
0 & \text{when X has unbounded variation,}
\end{cases}\\ 
\end{split}
\end{equation*}
where $\drift := \gamma+\int^{1}_0 z \Pi(\mathrm{d}z) > 0$ is the drift of $X$. \newline
We also define another scale function $Z_{q}(x,\theta )$ by 
\begin{equation}  \label{eq:zqscale2}
Z_{q}(x,\theta )=\mathrm{e}^{\theta x}\left( 1-\psi_q (\theta ) \int_{0}^{x}%
\mathrm{e}^{-\theta y}W_{q}(y)\mathrm{d}y\right) ,\quad x\geq 0,
\end{equation}%
and $Z_{q}(x,\theta )=\mathrm{e}^{\theta x}$ for $x<0$. For $\theta=0$, 
\begin{equation}  \label{eq:zqscale}
Z_{q}(x,0)=Z_{q}(x) = 1 + q \int_0^x W_{q}(y)\mathrm{d }y, \quad x \in 
\mathbb{R}.
\end{equation}
For $\theta\geq \Phi_q$, using \eqref{def_scale}, the scale function $%
Z_{q}(x,\theta)$ can be rewritten as 
\begin{equation}  \label{Zv2}
Z_{q}(x,\theta )=\psi_q (\theta ) \int_{0}^{\infty}\mathrm{e}^{-\theta
y}W_{q}(x+y)\mathrm{d}y ,\quad x\geq 0.
\end{equation}
We recall the \textit{delayed $q$-scale function of $X$} introduced by
Loeffen et al. \cite{loeffenetal2017} defined as 
\begin{equation}
\Lambda ^{\left( q\right) }\left( x,r\right) =\int_{0}^{\infty }W_{q}\left(
x+z\right) \frac{z}{r}\mathbb{P}\left( X_{r}\in \mathrm{d}z\right) ,
\label{DS2}
\end{equation}
and we write $\Lambda =\Lambda ^{(0)}$ when $q=0$. Note that we can show 
\begin{equation}  \label{L11}
\Lambda ^{\left(q\right)}\left(0,r\right)= \mathrm{e}^{qr}.
\end{equation}
We also denote the partial derivative of $\Lambda ^{\left( q\right) }$ with
respect to $x$ by 
\begin{equation*}
\Lambda ^{(q)^{\prime }}(x,r)=\frac{\partial \Lambda ^{(q)}}{\partial x}%
(x,r)=\int_{0}^{\infty }W_{q}^{\prime }\left( x+z\right) \frac{z}{r}\mathbb{P%
}\left( X_{r}\in \mathrm{d}z\right) ,
\end{equation*}
For $x\leq 0$, we also have 
\begin{equation}  \label{lambdaiden}
\Lambda ^{\prime}\left(x,r\right) =\Lambda ^{\left( q\right) \prime }\left(
x,r\right)-q\int_{0}^{r }\Lambda ^{\left( q\right) \prime }\left( x,s\right) 
\mathrm{d}s-qW_q (x).
\end{equation}

Also, we recall the following identity stated in \cite{loeffenetal2014} and
related to scale functions : for $p,q,x\geq 0$, 
\begin{equation}  \label{convsln}
(s-p) \int_0^x W_{p}(x-y) W_{s}(y) \mathrm{d}y = W_{s}(x) - W_{p}(x).
\end{equation}
Finally, we recall Kendall's identity that provides the distribution off the
first upward crossing of a specific level (see \cite[Corollary VII.3]%
{bertoin1996}): on $(0,\infty) \times (0,\infty)$, we have 
\begin{equation}  \label{eq:Kendall}
r \mathbb{P}(\tau_z^+ \in \mathrm{d}r) \mathrm{d}z = z \mathbb{P}(X_r \in 
\mathrm{d}z) \mathrm{d}r .
\end{equation}
We refer the reader to \cite{kyprianou2014} for more details on spectrally
negative L\'{e}vy processes and fluctuation identities. More examples and
numerical computations related to scale functions can be found in e.g., \cite%
{kuznetsovetal2012} and \cite{surya2008}.

\subsection{Main results}

\label{section3}

\subsubsection{Distribution of occupation times}

This subsection presents our main result for the density of the occupation
time $\mathcal{O}^{X}_{\mathrm{e}_{\lambda }}$ for a spectrally negative L%
\'{e}vy process $X$, where, throughout this paper, $\mathrm{e}_{\lambda }$
denotes an exponential random variable with rate $\lambda >0$ that is
independent of the process $X$. We first give the following lemma on the
Laplace transform of $\mathcal{O}^{X}_{\mathrm{e} _{\lambda }}$. Note that this result is essentially known in the literature (one may integrate the expression of $\mathbb{E}_{x}\left[ \mathrm{e}^{-q\mathcal{%
O}^{X}_{\mathrm{e} _{\lambda }}};X_{\mathrm{e}_{\lambda }}\in \mathrm{d} y\right] $ (with respect to $y)$, see, e.g.,  Corollary 1 of Gu\'{e}rin and Renaud \cite{guerinrenaud2015}), but an alternative and shorter proof is provided her.
\begin{lemma}
For $\lambda ,q>0$ and $x\in \mathbb{R}$,  
\begin{equation}
\mathbb{E}_{x}\left[ \mathrm{e}^{-q\mathcal{O}^{X}_{\mathrm{e}_{\lambda }}}%
\right] =\lambda \dfrac{\left( \Phi _{q+\lambda }-\Phi _{\lambda }\right) }{%
\left( \lambda +q\right) \Phi _{\lambda }}Z_{\lambda }\left( x,\Phi
_{\lambda +q}\right) - \dfrac{qZ_{\lambda }(x)}{q+\lambda }+1.  \label{LT}
\end{equation}
\end{lemma}

\begin{proof}
Using Proposition $3.4$ in Gu\'{e}rin and Renaud \cite{guerinrenaud2015}, one can deduce that  
\begin{equation*}
\mathbb{E}_{x}\left[ \mathrm{e}^{-q\mathcal{O}^{X}_{\mathrm{e}_{\lambda }}}%
\right] =\mathbb{P} _{x}\left( \mathcal{O}^{X}_{\mathrm{e}_{\lambda }}<%
\mathrm{e}_{q}\right) =\mathbb{P}_{x}\left( \kappa ^{q}>\mathrm{e} _{\lambda
}\right) =1-\mathbb{E}_{x}\left[ \mathrm{e}^{-\lambda \kappa ^{q}} \right] ,
\end{equation*}
where $\kappa ^{q}$ is the time of Parisian ruin with exponential delays 
defined as  
\begin{equation}
\kappa ^{q}=\inf \left\{ t>0:t-g_{t}>\mathrm{e}_{q}^{g_{t}}\right\} ,
\label{kappaq}
\end{equation}
where $g_{t}=\sup \left\{ 0\leq s\leq t\colon X_{s}\geq 0\right\} $ and $%
\mathrm{e} _{q}^{g_{t}}$ is an exponentially distributed random variable
with rate $q>0$ (independent of $X$). The Laplace transform of $\kappa ^{q}$
can be  extracted from Bardoux et al. \cite{baurdoux_et_al_2015} (see also
Albrecher  et al. \cite{albrecheretal2016}) and it is given by  
\begin{equation*}
\mathbb{E}_{x}\left[ \mathrm{e}^{-\lambda \kappa ^{q}}\right] =\dfrac{
qZ_{\lambda }(x)}{q+\lambda }-\lambda \dfrac{ \Phi _{q+\lambda }-\Phi
_{\lambda }}{\left( \lambda +q\right) \Phi _{\lambda }}Z_{\lambda
}\left( x,\Phi _{\lambda +q}\right) .
\end{equation*}
Therefore,
\begin{equation*}
\mathbb{E}_{x}\left[ \mathrm{e}^{-q\mathcal{O}^{X}_{\mathrm{e}_{\lambda }}}%
\right] =1-\mathbb{E }_{x}\left[ \mathrm{e}^{-\lambda \kappa ^{q}}\right]
=\lambda \dfrac{\Phi _{q+\lambda }-\Phi _{\lambda } }{\left(
\lambda +q\right) \Phi _{\lambda }}Z_{\lambda }\left( x,\Phi _{\lambda
+q}\right) -\dfrac{ qZ_{\lambda }(x)}{q+\lambda }+1.
\end{equation*}
\end{proof}

The following main theorem presents the density of $\mathcal{O}^{X}_{\mathrm{%
e} _{\lambda }}$, which is derived from (\ref{LT}) using the Laplace
inversion technique.

\begin{theorem}
\label{maintheo} For $\lambda >0$, $x \in \mathbb{R} $ and $y \geq 0$, 
\begin{eqnarray}  \label{mainres}
\mathbb{P}_{x}\left( \mathcal{O}^{X}_{\mathrm{e}_{\lambda }}\in \mathrm{d}%
y\right) &=& \left( 1-\left( Z_{\lambda }(x)-\frac{\lambda }{\Phi_\lambda }%
W_{\lambda }\left( x\right) \right) \right)\delta _{0}\left( \mathrm{d}%
y\right)  \notag \\
&&+\lambda \mathrm{e}^{-\lambda y}\left( \mathcal{B}^{\left( \lambda \right)
}\left( x,y\right) -\lambda \int_{0}^{y}\mathcal{B}^{\left( \lambda \right)
}\left( x,s\right) \mathrm{d}s \right) \mathrm{d} y  \notag \\
&&+\lambda \mathrm{e}^{-\lambda y} \left( Z_{\lambda}(x)-\frac{\lambda }{%
\Phi_\lambda }W_{\lambda }\left( x\right)\right)\mathrm{d} y,
\end{eqnarray}
where 
\begin{eqnarray*}
\mathcal{B}^{\left(\lambda\right) }\left(x,s\right) =\dfrac{\Lambda^{\left(
\lambda \right) \prime }\left( x,s\right)}{\Phi_\lambda }- \Lambda^{\left(
\lambda \right)}\left( x,s\right),
\end{eqnarray*}
and $\delta_0(\cdot)$ is the Dirac mass at $0$. In particular, when $x=0$,
equation \eqref{mainres} reduces to 

\begin{eqnarray}  \label{mainres0}
\mathbb{P}\left( \mathcal{O}^{X}_{\mathrm{e}_{\lambda }}\in \mathrm{d}%
y\right) = \frac{\lambda }{\Phi_\lambda} \left( W_{\lambda }\left(
0\right)\delta _{0}\left( \mathrm{d}y\right) +\mathrm{e}^{-\lambda y}\Lambda
^{\prime }\left( 0,y\right)\mathrm{d} y\right) .
\end{eqnarray}
\end{theorem}

\begin{proof}
Given that 
\begin{equation*}
\mathbb{P}_{x}\left( \mathcal{O}_{\mathrm{e}_{\lambda }}^{X}=0\right) =%
\mathbb{P}_{x}\left( \tau _{0}^{-}>\mathrm{e}_{\lambda }\right)
=1-Z_{\lambda }\left( x\right) +\frac{\lambda }{\Phi _{\lambda }}W_{\lambda
}\left( x\right) \text{,} 
\end{equation*}%
which is $0$ if $x<0$, we first rewrite \eqref{LT} as 
\begin{eqnarray*}
\mathbb{E}_{x}\left[\mathrm{e}^{-q\mathcal{O}_{\mathrm{e}_{\lambda }}^{X}}%
\right] &=&\left\{ 1-Z_{\lambda }\left( x\right) +\frac{\lambda }{\Phi
_{\lambda }}W_{\lambda }\left( x\right) \right\} +\frac{\lambda }{\lambda +q}%
\frac{\Phi _{\lambda +q}-\Phi _{\lambda }}{\Phi _{\lambda }}Z_{\lambda
}\left( x,\Phi _{\lambda +q}\right) \\
&&+\frac{\lambda }{\lambda +q}Z_{\lambda }\left( x\right) -\frac{\lambda }{%
\Phi _{\lambda }}W_{\lambda }\left( x\right) .
\end{eqnarray*}
By simple manipulations, the above expression can also be rewritten as 
\begin{eqnarray}  \label{lapinv}
\mathbb{E}_{x}\left[ \mathrm{e}^{-q\mathcal{O}_{\mathrm{e}_{\lambda }}^{X}}%
\right] &=& 1-Z_{\lambda }\left( x\right) +\frac{\lambda }{\Phi _{\lambda }}%
W_{\lambda }\left( x\right)  \notag \\
&&+\frac{\lambda }{\lambda +q}\left( \frac{\Phi _{\lambda +q}Z_{\lambda
}\left( x,\Phi _{\lambda +q}\right) -qW_{\lambda }\left( x\right) }{\Phi
_{\lambda }}-Z_{\lambda }\left( x,\Phi _{\lambda +q}\right) +Z_{\lambda
}\left( x\right) -\frac{\lambda }{\Phi _{\lambda }}W_{\lambda }\left(
x\right) \right)  \notag \\
&=& 1-Z_{\lambda }\left( x\right) +\frac{\lambda }{\Phi _{\lambda }}%
W_{\lambda }\left( x\right) +\frac{\lambda }{\lambda +q}\left( Z_{\lambda
}\left( x\right) -\frac{\lambda }{\Phi _{\lambda }}W_{\lambda }\left(
x\right) \right)  \notag \\
&&+\lambda \left( 1-\frac{\lambda }{\lambda +q}\right) \left( \frac{\Phi
_{\lambda +q}Z_{\lambda }\left( x,\Phi _{\lambda +q}\right) -qW_{\lambda
}\left( x\right) }{q\Phi _{\lambda }}-\frac{Z_{\lambda }\left( x,\Phi
_{\lambda +q}\right) }{q}\right) \text{.}
\end{eqnarray}
Using the following identities 
\begin{equation}  \label{id3}
\frac{Z_{\lambda }\left( x,\Phi_{\lambda +q}\right) }{q}=\int_{0}^{\infty }%
\mathrm{e}^{-qy}\left( \mathrm{e}^{-\lambda y}\Lambda ^{\left( \lambda
\right) }\left( x,y\right) \right) \mathrm{d}y,
\end{equation}
and 
\begin{equation}  \label{id1}
\frac{\Phi _{\lambda +q}Z_{\lambda }\left( x,\Phi _{\lambda
+q}\right)-qW_{\lambda }\left( x\right) }{q}=\int_{0}^{\infty }\mathrm{e}%
^{-qy}\left( \mathrm{e}^{-\lambda y}\Lambda ^{\left( \lambda \right) \prime
}\left( x,y\right) \right) \text{d}y,
\end{equation}
which can be proved using Kendall's identity \eqref{eq:Kendall} and
Tonelli's Theorem, leads to 
\begin{eqnarray*}
\mathbb{E}_{x}\left[ \mathrm{e}^{-q\mathcal{O}_{\mathrm{e}_{\lambda }}^{X}}%
\right] &=&\left\{ 1-Z_{\lambda }\left( x\right) +\frac{\lambda }{\Phi
_{\lambda }}W_{\lambda }\left( x\right) \right\} +\frac{\lambda }{\lambda +q}%
\left( Z_{\lambda }\left( x\right) -\frac{\lambda }{\Phi _{\lambda }}%
W_{\lambda }\left( x\right) \right) \\
&&+\lambda \left( 1-\frac{\lambda }{\lambda +q}\right) \int_{0}^{\infty }%
\mathrm{e}^{-qy}\left\{ \mathrm{e}^{-\lambda y}\mathcal{B}^{\left( \lambda
\right) }\left( x,y\right) \right\} \text{d}y\text{.}
\end{eqnarray*}%
Hence, by Laplace inversion, we obtain 
\begin{eqnarray*}
\mathbb{P}_{x}\left( \mathcal{O}_{\mathrm{e}_{\lambda }}^{X}\in \text{d}%
y\right) &=&\left( 1-Z_{\lambda }\left( x\right) +\frac{\lambda }{\Phi
_{\lambda }}W_{\lambda }\left( x\right) \right) \delta _{0}\left( \text{d}%
y\right) \\
&&+\lambda \mathrm{e}^{-\lambda y}\left( Z_{\lambda }\left( x\right) -\frac{%
\lambda }{\Phi _{\lambda }}W_{\lambda }\left( x\right) \right) \text{d}y \\
&&+\lambda\mathrm{e}^{-\lambda y} \left(\mathcal{B}^{\left( \lambda \right)
}\left( x,y\right)-\lambda \int_{0}^{y}\mathcal{B}^{\left( \lambda \right)
}\left(x,s\right) \text{d}s \right)\text{d}y.
\end{eqnarray*}
Equation \eqref{mainres} reduces to equation  \eqref{mainres0} when $x=0$.
This ends the proof.
\end{proof}

\bigskip

Note that the expression of $\mathbb{P}_{x}\left( \mathcal{O}_{\mathrm{e}%
_{\lambda }}^{X}\in \text{d}y\right)$ in \eqref{mainres} only relies on the
scale function $W_{\lambda}(x)$ and the law $\mathbb{P}\left( X_{r}\in 
\mathrm{d}z\right)$. Hence, one can obtain a
closed-form expression for the density of $\mathcal{O}_{\mathrm{e}_{\lambda
}}^{X}$ as long as $W_{\lambda}(x)$ and $\mathbb{P}\left( X_{r}\in \mathrm{d}%
z\right)$ are explicit. In Section \ref{sectionexamp}, we will provide a few
examples of the underlying process $X$ such that these two quantities are
explicit.

\begin{remark}
To better understand the formula \eqref{mainres}, we can rewrite it as
\begin{equation*}
\mathbb{P}_{x}\left( \mathcal{O}^{X}_{\mathrm{e}_{\lambda }}\in \mathrm{d}%
y\right) = \mathbb{P}_{x}\left( \tau _{0}^{-}>\mathrm{e}_{\lambda }\right)
\delta _{0}\left( \mathrm{d}y\right) +\mathbb{P}_{x}\left( \mathcal{O}^{X}_{%
\mathrm{e} _{\lambda }}\in \mathrm{d}y,\tau _{0}^{-}<\mathrm{e}_{\lambda
}\right) ,
\end{equation*}
where  
\begin{equation*}
\mathbb{P}_{x}\left( \tau _{0}^{-}>\mathrm{e}_{\lambda }\right) =1-\left(
Z_{\lambda }(x)-\frac{\lambda }{\Phi _{\lambda }}W_{\lambda }\left( x\right)
\right) ,
\end{equation*}
and  
\begin{eqnarray*}
\mathbb{P}_{x}\left( \mathcal{O}^{X}_{\mathrm{e}_{\lambda }}\in \mathrm{d}%
y,\tau _{0}^{-}<\mathrm{e}_{\lambda }\right) &=&\lambda \mathrm{e}^{-\lambda
y}\left( Z_{\lambda }\left( x\right) -\frac{\lambda }{\Phi _{\lambda }}%
W_{\lambda }\left( x\right) \right) \text{d}y \\
&&+\lambda \mathrm{e}^{-\lambda y}\left( \mathcal{B}^{\left( \lambda \right)
}\left( x,y\right) -\lambda \int_{0}^{y}\mathcal{B}^{\left( \lambda \right)
}\left( x,s\right) \mathrm{d}s\right) \mathrm{d} y.
\end{eqnarray*}
\end{remark}

\begin{remark}
The time horizon in Theorem \eqref{maintheo} can be easily extended to a 
hypoexponential distribution (also called generalized Erlang distribution). The class of hypoexponential distributions is known to be dense within the class
of continuous nonnegative distributions in terms of weak convergence; see, 
e.g., Botta and Harris \cite{botta1986approximation}. Let $\tilde{\me}_{n}$ be
a hypoexponential-distributed time horizon given by $\tilde{\me}_{n}=\mathrm{e%
} _{1}+\mathrm{e}_{2}+\cdots +\mathrm{e}_{n}$, where $\mathrm{e}_{1},\mathrm{%
e} _{2},\ldots ,\mathrm{e}_{n}$ are mutually independent and exponentially 
distributed with distinct means $1/\lambda _{1},1/\lambda _{2},\ldots ,1/\lambda _{n}$
, respectively. Since the density of $\tilde{e}_{n}$ is given by 
\begin{equation*}
\mathbb{P}\left( \tilde{\me}_{n}\in \mathrm{d}t\right)
=\sum_{k=1}^{n}a_{k}^{n}\lambda _{k}\mathrm{e}^{-\lambda _{k}t}\mathrm{d}t,
\end{equation*}%
where $\lambda _{k}>0$ and $a_{k}^{n}=\prod\limits_{j=1,j\neq k}^{n}\frac{\lambda _{j}}{\lambda_{j}-\lambda _{k}}$ with $\sum_{k=1}^{n}a_{k}^{n}=1$, it follows that
\begin{eqnarray*}
\mathbb{P}_{x}\left( \mathcal{O}^{X}_{\tilde{\mathrm{e}}_{n}}\in \mathrm{d}%
y\right) &=&\int_{0}^{\infty }\sum_{k=1}^{n}a_{k}^{n}\lambda _{k}\mathrm{e}%
^{-\lambda _{k}t}\mathbb{P}_{x}\left( \mathcal{O}^{X}_{t}\in \mathrm{d}%
y\right) \mathrm{d}t \\
&=&\sum_{k=1}^{n}a_{k}^{n}\int_{0}^{\infty }\lambda _{k}\mathrm{e}^{-\lambda_{k}t}\mathbb{P}_{x}\left( \mathcal{O}^{X}_{t}\in \mathrm{d}y\right) \mathrm{%
d}t \\
&=&\sum_{k=1}^{n}a_{k}^{n}\mathbb{P}_{x}\left( \mathcal{O}^{X}_{\mathrm{e}%
_{\lambda _{k}}}\in \mathrm{d}y\right) .
\end{eqnarray*}
For instance, by setting 
\begin{equation*}
\lambda _{k}=\frac{n(n+1)}{2k\mathbb{E}\left[ \tilde{\me}_{n}\right] },
\end{equation*}
the variance of $\tilde{\me}_{n}$ became $\dfrac{2n(2n+1)}{3n(n+1)}k\left( 
\mathbb{E}\left[ \tilde{\me}_{n}\right] \right) ^{2}$ which converges to $0$
when $n\rightarrow \infty $. Thus, it is possible to use the Erlangization
method to approximate the fixed time horizon case (see Klugman et al. \cite%
{klugmanetal2012}).
\end{remark}

Letting $\lambda\rightarrow0$ in \eqref{mainres}, we obtain the following
expression for the distribution of $\mathcal{O}^{X}_{\infty}$, the
occupation time $X$ stays below level $0$ up to infinity.

\begin{corol}
\label{corol} For $x \in \mathbb{R} $, $y \geq 0$ and $\mathbb{E}[X_{1}]>0$, 
\begin{equation}  \label{corolinf}
\mathbb{P}_{x}\left( \mathcal{O}^{X}_{\infty} \in \mathrm{d} y \right)= 
\mathbb{E}[X_{1}] \left( W(x)\delta_0(\mathrm{d} y)+ \Lambda^{\prime}(x,y)%
\mathrm{d} y \right).
\end{equation}
\end{corol}

Again we can rewrite (\ref{corolinf}) as 
\begin{equation*}
\mathbb{P}_{x}\left( \mathcal{O}^{X}_{\infty }\in \mathrm{d} y\right) =%
\mathbb{P}_{x}\left( \tau _{0}^{-}=\infty \right) \delta _{0}\left( \mathrm{d%
}y\right) +\mathbb{P}_{x}\left( \mathcal{O}^{X}_{\infty }\in \mathrm{d}%
y,\tau _{0}^{-}<\infty \right) ,
\end{equation*}%
where 
\begin{equation*}
\mathbb{P}_{x}\left( \tau _{0}^{-}=\infty \right) =\mathbb{E}[X_{1}]W(x),
\end{equation*}%
and 
\begin{equation*}
\mathbb{P}_{x}\left( \mathcal{O}^{X}_{\infty }\in \mathrm{d}y,\tau
_{0}^{-}<\infty \right) =\mathbb{E}[X_{1}]\Lambda ^{\prime }(x,y)\mathrm{d}
y.
\end{equation*}

\subsection{Applications}

\label{section2} This section is devoted to introduce a few applications of
Theorem \ref{maintheo}.

\subsubsection{Future drawdown}
Drawdown is used as a dynamic risk metric to measure the magnitude of the 
decline of insurance surplus from its maximum. Interested readers are
referred to Zhang \cite{hongzhong2018stochastic} for more theoretical
results and applications of drawdown in insurance and finance. Recently,
Baurdoux et al. \cite{baurdouxetal2017} introduced the future drawdown
extreme defined as  
\begin{equation*}
\bar{D}_{s,t}=\sup_{0\leq u\leq s}\inf_{u\leq w\leq t+s}\left(
X_{w}-X_{u}\right) ,
\end{equation*}
where $s,t>0$. The infinite-horizon version is denoted by  
\begin{equation*}
\bar{D}_{s}=\lim_{t\rightarrow \infty }\bar{D}_{s,t}=\sup_{0\leq u\leq
s}\inf_{w\geq u}\left( X_{w}-X_{u}\right) .
\end{equation*}
From Corollary 5.2 (ii) of Baurdoux et al. \cite{baurdouxetal2017}, we have  
\begin{equation}
\mathbb{P}\left( -\bar{D}_{\mathrm{e}_{q}}<x\right) =\mathbb{E}[X_{1}]\frac{%
\Phi _{q}}{q}Z\left( x,\Phi _{q}\right)=\mathbb{E}_{x}\left[ \mathrm{e}^{-q%
\mathcal{O}^{X} _{\infty }}\right],  \label{drd}
\end{equation}
where the last equality is due to Corollary 1 of Landriault et al. \cite%
{landriaultetal2011}. By \eqref{corolinf}, we conclude that  
\begin{equation*}
\mathbb{P}\left( -\bar{D}_{s}<x\right) =\mathbb{P}_{x}\left( \mathcal{O}%
^{X}_{\infty }<s\right) = \mathbb{E}[X_{1}]\left( W(x)+\int_{0}^{s}\Lambda
^{\prime }(x,y)\mathrm{d} y\right) .
\end{equation*}
\subsubsection{Inverse occupation time}
The occupation time $\mathcal{O}^{X}_{t}$ certainly consists of some
information on how long a surplus process may stay in the red zone up to
time $t$. But it fails to provide a solvency early warning mechanism (in the
form of a stopping time or others) that the insurer can act on in periods of
financial distress. This motivates us to consider the \emph{inverse
occupation time}, that is the first time the accumulated duration of all
periods of financial distress (periods in which the risk process is below
the solvency threshold level) exceeds a deterministic tolerance level.
Specifically, the inverse occupation time with parameter $r>0$ is defined as 
\begin{equation*}
\sigma _{r}=\inf \left\{ t>0\colon \mathcal{O}^{X}_{t}>r\right\} .
\end{equation*}%
The stopping time $\sigma _{r}$ is deemed to occur at the first time the
process $X$ cumulatively stays below level $0$ in excess of $r$. Here, the
parameter $r$ can represent the insurer's tolerance level for the surplus
process to cumulatively stay below threshold $0$. Note that, in actuarial
ruin terminology, the inverse occupation time is also known as the \textit{\
cumulative Parisian ruin} time (see Gu\'{e}rin and Renaud \cite%
{guerinrenaud2015}).

The finite-time probability of inverse occupation time is given by 
\begin{equation}
\mathbb{P}_{x}\left( \sigma _{r}\leq t\right) =\mathbb{P}_{x}\left( \mathcal{%
O}^{X}_{t}>r\right) ,  \label{cumu2}
\end{equation}%
while in the infinite-time horizon case 
\begin{equation}
\mathbb{P}_{x}\left( \sigma _{r}<\infty \right) =\mathbb{P}_{x}\left( 
\mathcal{O}^{X}_{\infty }>r\right) .  \label{cumu1}
\end{equation}%
The Laplace transform of $\sigma _{r}$ is given by 
\begin{equation}
\mathbb{E}_{x}\left[ \mathrm{e}^{-\lambda \sigma _{r}}\right] =\mathbb{P}%
_{x}\left( \sigma _{r}<\mathrm{e}_{\lambda }\right) =\mathbb{P}_{x}\left( 
\mathcal{O}^{X}_{\mathrm{e}_{\lambda }}>r\right) ,  \label{cumu}
\end{equation}%
which can be readily obtained from Theorem \ref{maintheo} as below.

\begin{theorem}
\label{maintheoParisian} For $r,\lambda >0$ and $x \in \mathbb{R} $,  
\begin{eqnarray}  \label{Lapcumu}
\mathbb{E}_x\left[ \mathrm{e}^{-\lambda \sigma _{r}}\right] &=&\mathrm{e}
^{-\lambda r}\left( Z_{\lambda }(x)-\frac{\lambda }{\Phi_\lambda} W_{\lambda
}\left( x\right) \right)  \notag \\
&&-\lambda\int_{0}^{r}\mathrm{e}^{-\lambda u}\left( \mathcal{B}^{\left(
\lambda \right) }\left( x,u\right) -\lambda \int_{0}^{u} \mathcal{B}^{\left(
\lambda \right) }\left( x,s\right) \mathrm{d}s\right) \mathrm{d}u.
\end{eqnarray}
\end{theorem}

Using \eqref{corolinf}, we obtain the following expression for the
probability that the inverse occupation time ever occurs.


\begin{corol}
\label{LaplaceTr}  For $r>0$, $x \in \mathbb{R} $ and $\mathbb{E}[X_1]>0$,  
\begin{eqnarray}  \label{cumuParruin}
\mathbb{P}_x\left( \sigma _{r} < \infty \right) &=&1-\mathbb{E}[X_1]\left(
W(x)+\int_{0}^{r} \Lambda^{ \prime }\left( x,s\right)\mathrm{d} s\right).
\end{eqnarray}
\end{corol}

Finally, we point out that \eqref{cumuParruin} reduces to $\mathbb{P}\left(
\tau _{0}^{-}<\infty \right) =1-\mathbb{E}[X_{1}]W\left( x\right) $ when $%
r\rightarrow 0$ and using also the fact that the stopping time $\sigma _{r}$
converges $\mathbb{P}_{x}-$a.s. to the time of classical ruin $\tau _{0}^{-}$
(see Proposition $3.3$ in \cite{guerinrenaud2015}).

\subsubsection{Parisian ruin with exponential delay}

Another type of ruin in actuarial science with strong ties to the
distribution of $\mathcal{O}^{X}_{t}$ is the time of Parisian ruin with
exponential delays $\kappa ^{q}$ defined in \eqref{kappaq}. An expression
for the probability of Parisian ruin with exponential delays was first given
in \cite{landriaultetal2011} through the relation between the occupation
time $\mathcal{O}^{X}_{\infty }$ and $\kappa ^{q}$, that is, for $\mathbb{E}
[X_{1}]>0$, $q>0$ and $x\in \mathbb{R}$, 
\begin{equation}
\mathbb{P}_{x}\left( \kappa ^{q}<\infty \right) =1-\mathbb{E}_{x}\left[ 
\mathrm{e}^{-q\mathcal{O}^{X}_{\infty }}\right] =1-\mathbb{E}\left[ X_{1}%
\right] \frac{\Phi _{q}}{q}Z\left( x,\Phi _{q}\right) .  \label{Pruine1}
\end{equation}

We can readily recover (\ref{Pruine1}) using our result. Given that from
Proposition $3.4$ in \cite{guerinrenaud2015}, it is known that $\kappa ^{q}$
and $\sigma _{\mathrm{e}_{q}}$ have the same distribution. Replacing the
delay $r$ by an exponential random time $\mathrm{e}_{q}$ in %
\eqref{cumuParruin}, 
\begin{eqnarray*}
\mathbb{P}_{x}\left( \kappa ^{q}<\infty \right) &=&\mathbb{P}_{x}\left(
\sigma _{\mathrm{e}_{q}}<\infty \right) \\
&=&1-\mathbb{E}[X_{1}]\int_{0}^{\infty }q\mathrm{e}^{-qr}\left(
W(x)+\int_{0}^{r}\Lambda ^{\prime }\left( x,s\right) \mathrm{d}s\right) 
\mathrm{d}r \\
&=&1-\mathbb{E}[X_{1}]\left( W(x)+\int_{0}^{\infty }\mathrm{e}^{-qs}\Lambda
^{\prime }\left( x,s\right) \mathrm{d}s\right) \\
&=&1-\mathbb{E}[X_{1}]\dfrac{\Phi _{q}}{q}Z\left( x,\Phi _{q}\right) ,
\end{eqnarray*}%
where the last equality follows from identity \eqref{id1}.

\subsubsection{Last time at running maximum}
Denote the last time $X$ was at its peak by 
\begin{equation*}
G_{t}=\sup \left\{ s\leq t:X_{s}=\bar{X}_{s}\right\} ,
\end{equation*}%
where $\bar{X}_{t}=\sup_{s\leq t}X_{s}$ is the running maximum of $X$. The
quantity $t-G_t$ is so-called the duration of drawdown at time $t$ (see
Landriault et al. \cite{landriault2017magnitude}). We also denote the
occupation time of $X$ in the positive half-line by 
\begin{equation*}
\widehat{\mathcal{O}}_{t}^{X}=\int_{0}^{t}\mathbf{1}_{[0,\infty )}\left(
X_{s}\right) \mathrm{d}s.
\end{equation*}%
Since $X_0=0$, by the Sparre Andersen's identity (see Lemma VI.15 of Bertoin 
\cite{bertoin1996}), we know that, for every $t>0$, 
\begin{equation*}
\widehat{\mathcal{O}}_{t}^{X}\overset{\text{law}}{=}G_{t}.
\end{equation*}%
Then we have 
\begin{eqnarray*}
\mathbb{P}\left( \mathrm{e}_{\lambda }-G_{\mathrm{e}_{\lambda }}\in \mathrm{%
d }y\right) &=&\mathbb{P}\left( \mathrm{e}_{\lambda }-\widehat{\mathcal{O}}%
_{ \mathrm{e}_{\lambda }}^{X}\in \mathrm{d}y\right) \\
&=&\mathbb{P}\left( \mathcal{O}_{\mathrm{e}_{\lambda }}^{X}\in \mathrm{d}
y\right) \\
&=&\frac{\lambda }{\Phi _{\lambda }}\left( W_{\lambda }\left( 0\right)
\delta _{0}\left( \mathrm{d}y\right) +\mathrm{e}^{-\lambda y}\Lambda
^{\prime }\left( 0,y\right) \mathrm{d}y\right) ,
\end{eqnarray*}%
where the second last equality is due to \eqref{mainres0}.
\subsection{Examples\label{sectionexamp}}

This subsection is devoted to provide some examples of the spectrally
negative L\'{e}vy process $X$ for the main results in Theorem \ref{maintheo}%
, i.e., the law of $\mathcal{O}^{X}_{\mathrm{e}_{\lambda }}$. For cases of
Brownian risk process and Cram\'{e}r-Lundberg process with exponential
claims, we will obtain the law of $\mathcal{O}^{X}_{t}$ by a further
inversion. We assume $X_{0}=0$ in the following examples for simplicity.

\subsubsection{Brownian risk process}

Let $X_{t}=\mu t+\sigma B_{t},$ where $\mu >0$, $\sigma >0$, and $%
\{B_{t}\}_{t\geq 0}$ is a standard Brownian motion. For this process, the
scale function and the right-inverse of the Laplace exponent are given by 
\begin{equation*}
W(x)=\frac{1}{\mu }\left( 1-\mathrm{e}^{-2\mu x/\sigma ^{2}}\right) ,\quad
x\geq 0,
\end{equation*}%
and 
\begin{equation*}
\Phi _{\lambda }=\left( \sqrt{\mu ^{2}+2\lambda \sigma ^{2}}-\mu \right)
\sigma ^{-2},\quad \lambda >0,
\end{equation*}%
respectively. Also, since $X_{s}$ has a normal distribution with mean $\mu s$
and variance $s\sigma ^{2}$, 
\begin{equation*}
\Lambda \left( x,s\right) =\left( \frac{\sigma \mathrm{e}^{-\frac{\mu ^{2}s}{
2\sigma ^{2}}}}{\mu \sqrt{r2\pi }}+\mathcal{N}\left( \frac{\mu \sqrt{s}}{
\sigma }\right) \right) \left( 1-\mathrm{e}^{-\frac{2\mu }{\sigma ^{2}}
x}\right) +\mathrm{e}^{-\frac{2\mu }{\sigma ^{2}}x},
\end{equation*}%
and consequently, 
\begin{equation}
\Lambda ^{\prime }\left( x,s\right) =\frac{2}{\sigma ^{2}}\mathrm{e}^{-\frac{
2\mu }{\sigma ^{2}}x}\left( \frac{\sigma \mathrm{e}^{-\frac{\mu ^{2}s}{
2\sigma ^{2}}}}{\sqrt{r2\pi }}-\mu \mathcal{\bar{N}}\left( \frac{\mu \sqrt{s}%
}{\sigma }\right) \right) ,  \label{star}
\end{equation}%
where $\mathcal{N}=1-\bar{\mathcal{N}}$ is the cumulative distribution
function of the standard normal distribution. One can easily check that 
\begin{equation}
\frac{\mathrm{e}^{-\lambda s}}{\Phi _{\lambda }}=\int_{0}^{\infty }\mathrm{e}%
^{-\lambda t}\left( \mu +\frac{\sigma \mathrm{e}^{-(\mu ^{2}/2\sigma
^{2})\left( t-s\right) }}{\sqrt{2\pi \left( t-s\right) }}-\mu \bar{\mathcal{N%
}}\left( \frac{\mu \sqrt{t-s}}{\sigma }\right) \right) \mathrm{d}t.
\label{inverse}
\end{equation}%
Since $X$ has paths of unbounded variation (i.e., $W_{\lambda }(0)=0$) and
from \eqref{star} at $x=0$, we have 
\begin{equation*}
\lambda ^{-1}\mathbb{P}\left( \mathcal{O}^{X}_{\mathrm{e}_{\lambda }}\in 
\mathrm{d} s\right) =\frac{ \mathrm{e}^{-\lambda s}}{\Phi _{\lambda }}%
\Lambda ^{\prime }\left( 0,s\right) =\frac{\mathrm{e}^{-\lambda s}}{\Phi
_{\lambda }}\frac{2}{\sigma ^{2}}\left\{ \frac{\sigma \mathrm{e}^{-(\mu
^{2}/2\sigma ^{2})s}}{\sqrt{2\pi s}}-\mu \bar{\mathcal{N}}\left( \frac{\mu 
\sqrt{s}}{\sigma }\right) \right\} \mathrm{d}s.
\end{equation*}%
Using Laplace inversion, we finally obtain 
\begin{eqnarray*}
\mathbb{P}\left( \mathcal{O}^{X}_{t}\in \mathrm{d} s\right) &=&\frac{2}{%
\sigma ^{2}}\left\{ \frac{\sigma \mathrm{e}^{-(\mu ^{2}/2\sigma ^{2})s}}{%
\sqrt{2\pi s}}-\mu \bar{ \mathcal{N}}\left( \frac{\mu \sqrt{s}}{\sigma }%
\right) \right\} \\
&&\times \left\{ \mu +\frac{\sigma \mathrm{e}^{-(\mu ^{2}/2\sigma
^{2})\left( t-s\right) }}{\sqrt{2\pi \left( t-s\right) }}-\mu \bar{\mathcal{%
N }}\left( \frac{\mu \sqrt{t-s}}{\sigma }\right) \right\} \mathrm{d}s,
\end{eqnarray*}%
which is consistent with the result in Akahori \cite{akahori1995}. But we
point out that our approach is under a more general framework of spectrally
negative L\'{e}vy process, while Akahori \cite{akahori1995} uses the
specific Feynman--Kac formula for Brownian motions. In particular, letting $%
\sigma =1$, $\mu =0$, and integrating the law of $\mathcal{O}^{X}_{t}$ over $%
[r,\infty ) $, one obtains the famous Paul L\'{e}vy's arcsine law, that is, 
\begin{equation*}
\mathbb{P}\left( \mathcal{O}^{X}_{t}>r\right) =1-\dfrac{2}{\pi }\arcsin
\left( \sqrt{ \dfrac{r}{t}}\right) ,\quad 0<r<t.
\end{equation*}

\subsubsection{Cram\'{e}r-Lundberg process with exponential claims}

Let $X$ be a Cram\'{e}r-Lundberg risk process with exponentially distributed
claims, i.e.,\ 
\begin{equation*}
X_{t}=c t-\sum_{i=1}^{N_{t}}C_{i},
\end{equation*}%
where $\{N_{t}\}_{t\geq 0}$ is a Poisson process with intensity $\eta >0$,
and $\{C_{1},C_{2},\dots \}$ are independent and exponentially distributed
random variables with parameter $\alpha $, also independent of $N$. The
scale function of $X$ is known to be 
\begin{equation*}
W(x)=\frac{1}{c-\eta /\alpha }\left( 1-\frac{\eta }{c\alpha }\mathrm{e}^{( 
\frac{\eta }{c}-\alpha )x}\right) ,
\end{equation*}%
and the right-inverse has the closed-form expression 
\begin{equation*}
\Phi _{\lambda }=\frac{1}{2c}\left( \lambda +\eta -c\alpha +\sqrt{\left(
\lambda +\eta -c\alpha \right) ^{2}+4c\alpha \lambda }\right) .
\end{equation*}%
As noted in Loeffen et al. \cite{loeffenetal2013}, we have%
\begin{equation*}
\mathbb{P}\left( \sum_{i=1}^{N_{s}}C_{i}\in \mathrm{d}y\right) =\mathrm{e}%
^{-\eta s}\left( \delta _{0}(\mathrm{d}y)+\mathrm{e}^{-\alpha
y}\sum_{m=0}^{\infty }\frac{(\alpha \eta s)^{m+1}}{m!(m+1)!}y^{m}\mathrm{d}%
y\right) ,
\end{equation*}%
and consequently 
\begin{eqnarray*}
\int_{0}^{\infty }z\mathbb{P}\left( X_{s}\in \mathrm{d}z\right)
&=&\int_{0}^{cs}z\mathrm{e}^{-\eta s}\left( \delta _{0}(cs-\mathrm{d}z)+ 
\mathrm{e}^{-\alpha (cs-z)}\sum_{m=0}^{\infty }\frac{(\alpha \eta s)^{m+1}}{
m!(m+1)!}(cs-z)^{m}\mathrm{d}z\right) \\
&=&\mathrm{e}^{-\eta s}\left( cs+\sum_{m=0}^{\infty }\frac{(\eta s)^{m+1}}{
m!(m+1)!}\left[ cs\Gamma (m+1,cs\alpha )-\frac{1}{\alpha }\Gamma
(m+2,cs\alpha )\right] \right) ,
\end{eqnarray*}%
where $\Gamma (a,x)=\int_{0}^{x}\mathrm{e}^{-t}t^{a-1}\mathrm{d}t$ is the
incomplete gamma function, and 
\begin{equation*}
\frac{\eta }{c\alpha }\int_{0}^{\infty }\mathrm{e}^{(\frac{\eta }{c}-\alpha
)z}z\mathbb{P}(X_{s}\in \mathrm{d}z)=\int_{0}^{\infty }z\mathbb{P}(X_{s}\in 
\mathrm{d}z)-(c-\eta /\alpha )s.
\end{equation*}%
Then, 
\begin{equation*}
\Lambda ^{\prime }\left( 0,s\right) =\frac{\alpha }{c}e^{-\eta s}\left(
c+\sum_{i=0}^{\infty }\frac{\eta ^{m+1}r^{m}}{m!\left( m+1\right) !}\left(
cs\Gamma \left( m+1,cs\alpha \right) -\frac{1}{\alpha }\Gamma \left(
m+2,cs\alpha \right) \right) \right) .
\end{equation*}%
and 
\begin{equation*}
\frac{1}{\Phi _{\lambda }c}=\frac{1}{\sqrt{\left( \lambda +\eta -c\alpha
\right) ^{2}+4c\alpha \lambda }-\left( c\alpha -\lambda -\eta \right) }.
\end{equation*}%
Since $X$ is of bounded variation paths (i.e., $W_{\lambda }(0)>0$), we have 
\begin{eqnarray*}
\lambda ^{-1}\mathbb{P}\left( \mathcal{O}^{X}_{\mathrm{e}_{\lambda }}\in 
\mathrm{d} s\right) = \frac{1}{\Phi _{\lambda }}W_{\lambda }\left( 0\right)
\delta _{0}\left( \mathrm{d}s\right) +\frac{\mathrm{e}^{-\lambda s}}{\Phi
_{\lambda }}\Lambda ^{\prime }\left( 0,s\right) \mathrm{d}s.
\end{eqnarray*}%
As shown in Gu\'{e}rin and Renuad \cite{guerinrenaud2015}, we have 
\begin{equation*}
\frac{1}{\Phi _{\lambda }c}=\int_{0}^{\infty }\mathrm{e}^{-\lambda t}a_{t}%
\mathrm{d}t,
\end{equation*}%
where 
\begin{equation*}
a_{t}=\left( 1-\frac{\eta }{c\alpha }\right) _{+}+\frac{2 \eta }{\pi }%
\mathrm{e}^{-\left( \eta +c\alpha \right) t}\int_{-1}^{1}\frac{\sqrt{1-u^{2}}%
\mathrm{e} ^{-2\sqrt{c\alpha \eta }tu}}{\eta +c\alpha +2\sqrt{c\alpha \eta }u%
}\mathrm{d}t.
\end{equation*}
Then, we also have 
\begin{equation*}
\frac{\mathrm{e}^{-\lambda s}}{\Phi _{\lambda }c}=\int_{0}^{\infty }\mathrm{e%
}^{-\lambda t}\left( a_{t-s}\mathbf{1}_{\left(0,t\right) } (s)\right)\mathrm{%
d}t,
\end{equation*}%
We then obtain the following expression for the distribution of the
occupation time $\mathcal{O}^{X}_{t}$ which is more compact than the one in 
\cite{guerinrenaud2015}: for $t>0$, 
\begin{equation*}
\mathbb{P}\left( \mathcal{O}^{X}_{t}\in \mathrm{d}s\right) =a_{t}\delta
_{0}\left( \mathrm{d} s\right) +c\Lambda ^{\prime }\left( 0,s\right) a_{t-s}%
\mathbf{1}_{\left( 0,t\right) }\left( s\right) \mathrm{d} s.
\end{equation*}
\\
For the next two examples, we aim to provide a characterization of $\mathcal{O}^{X}_{\mathrm{e}_{\lambda }}$ (rather than $\mathcal{O}^{X}_{t}$). As shown in Equation \eqref{mainres} (and \eqref{mainres0}), it is sufficient to identify the scale function $W_\lambda(x)$ and the density of $X_{t}$. For completeness, we recall known results pertaining to these quantities.
\subsubsection{Jump diffusion risk process with phase-type claims}

As a generalization of the previous two examples, we consider a jump diffusion
risk process with phase-type claims, that is, 
\begin{equation*}
X_{t}=c t+\sigma B_{t}-\sum_{i=1}^{N_{t}}C_{i},
\end{equation*}%
where $\sigma \geq 0$, $\{B_{t}\}_{t\geq 0}$ is a standard Brownian motion, $%
\{N_{t}\}_{t\geq 0}$ is a Poisson process with intensity $\eta >0$, and $%
\{C_{1},C_{2},\dots \}$ are independent random variables with common
phase-type distribution with with the minimal representation $(m,\mathbf{T},%
\boldsymbol{\alpha })$, i.e.\ its cumulative distribution function is given
by $F(x)=1-\boldsymbol{\alpha }\mathrm{e}^{\mathbf{T}x}\mathbf{1}$, where $%
\mathbf{T}$ is an $m\times m$ matrix of a continuous-time killed Markov
chain, its initial distribution is given by a simplex $\boldsymbol{\alpha }%
=[\alpha _{1},...,\alpha _{m}]$, and $\mathbf{1}$ denotes a column vector of
ones. All of the aforementioned objects are mutually independent (for more
details we refer to Egami and Yamazaki \cite{egamiyamazaki2014}).

The Laplace exponent of $X$ is known to be of the form
\begin{equation}
\psi (\lambda )=c\lambda +\frac{\sigma ^{2}\lambda ^{2}}{2}+\eta \left( 
\boldsymbol{\alpha }(\lambda \mathbf{I}-\mathbf{T})^{-1}\mathbf{t}-1\right) ,
\label{psi_Jump-diffusion}
\end{equation}%
where $\mathbf{t}=-\mathbf{T}\mathbf{1}$. Let $\rho _{j,\lambda}$ be the
roots with negative real parts of the equation $\theta \mapsto \psi
(\theta )=\lambda$. Since we assume the net profit condition $\mathbb{E}[X_1]>0$,
from Proposition 5.4 in Kuznetsov et al. \cite{kuznetsovetal2012}, we have
that the $\rho _{j,\lambda}$'s are distinct roots. Then, from Proposition 2.1 in 
\cite{egamiyamazaki2014}, we have 
\begin{equation*}
W_\lambda(x)=\frac{\me^{\Phi_\lambda x}}{\psi^\prime (\Phi_\lambda )}+\sum_{j\in \mathcal{I}_{\lambda }}A_{j,\lambda }%
\mathrm{e}^{\rho _{j,\lambda }x},
\end{equation*}%
where $A_{j,\lambda }=\frac{1}{\psi ^{\prime }(\rho _{j,\lambda })}$ and $\mathcal{I}_{\lambda }$
is the set of indices corresponding to the $\rho _{j,\lambda}$'s. 
Moreover, 
\begin{equation*}  \label{phase density}
\mathbb{P}(X_{t}\in \mathrm{d}z)=\mathrm{e}^{-\eta t}\sum_{k=0}^{\infty }%
\frac{(\eta t)^{k}}{k!}\int_{0}^{\infty }F^{\ast k}(\mathrm{d}y)\mathcal{N}%
\left( (\mathrm{d}z+y-ct)\sigma \sqrt{t}\right) ,
\end{equation*}%
where $\mathcal{N}$ is the cumulative distribution function of a standard
normal random variable, $F^{\ast k}$ is the $k$-th convolution of $F$ and
for $k=0$ we understand $F^{\ast 0}(\mathrm{d}y)=\delta _{0}(\mathrm{d}y)$.

\subsubsection{Stable risk process}

We suppose that $X$ is 
a spectrally negative $\alpha $-stable
process with $\alpha =3/2$. In this case, the Laplace exponent of $X$ is
given by $\psi (\lambda )=\lambda ^{3/2}$. Then, for $q,x\geq 0$, we
have 
\begin{equation*}
W_\lambda(x)=\dfrac{3 \sqrt{x}}{2}E^{\prime}_{3/2}(\lambda x^{3/2})
\end{equation*}%
where $E_{3/2}= \sum_{k\geq 0}z^{k}/\Gamma(1+3k/2) $ is the Mittag-Leffler function of order $3/2$. As noted in
Loeffen et al. \cite{loeffenetal2013}, we have 
\begin{equation*}
\mathbb{P}(X_{t}\in \mathrm{d}y)=\mathbb{P}(t^{2/3}X_{1}\in \mathrm{d}y)=%
\begin{cases}
\sqrt{\frac{3}{\pi }}t^{2/3}y^{-1}\mathrm{e}^{-u/2}W_{1/2,1/6}\left(
u\right) \mathrm{d}y & y>0, \\ 
-\frac{1}{2\sqrt{3\pi }}t^{2/3}y^{-1}\mathrm{e}^{u/2}W_{-1/2,1/6}\left(
u\right) \mathrm{d}y & y<0,%
\end{cases}%
\end{equation*}%
where $u=\frac{4}{27}t^{9/2}|y|^{3}$ and $W_{\kappa ,\mu }$ is Whittaker's
W-function (see Lebedev \cite{lebedev}). The density of $X_{t}$ is readily obtained by a simple change of
variable. 

\section{Occupation times of the Refracted L\'{e}vy process}

\label{section6} We now extend our results to a refracted spectrally
negative L\'{e}vy process $U=\{U_t\}_{t\geq 0}$ at level $0$ defined as 
\begin{equation*}
U_t = X_t - \delta \int^{t}_0 \mathbf{1}_{\{U_s >0\}} \mathrm{d}s , \quad t
\geq 0 , 
\end{equation*}
where $\delta \geq 0$ is the refraction parameter. As discussed in Kyprianou
and Loeffen \cite{kyprianouloeffen2010}, such process exists and it is a
skip-free upward strong Markov process. Above $0$, the surplus process $U$
evolves as $Y=\{Y_t=X_t-\delta t\}_{t\geq 0}$ for which the Laplace exponent
is given by 
\begin{equation*}
\lambda \mapsto \psi(\lambda) - \delta \lambda , 
\end{equation*}
with right-inverse $\varphi_q = \sup \{ \lambda \geq 0 : \psi(\lambda) -
\delta \lambda = q\}$. Then, for each $q \geq 0$, we define the scale
functions of $Y$, namely $\mathbb{W}_{q}$ and $\mathbb{Z}_{q}$, by 
\begin{equation*}
\int_0^{\infty} \mathrm{e}^{- \lambda y} \mathbb{W}_q (y) \mathrm{d}y = 
\frac{1}{\psi_q(\lambda) - \delta \lambda } , \quad \text{ $\lambda >
\varphi_q$}, 
\end{equation*}
and 
\begin{equation*}
\mathbb{Z}_{\delta,q }\left( x,\theta\right)=e^{\theta x}\left(
1-\left(\psi_q(\theta) - \delta \theta \right)\int_{0}^{x}\mathrm{e}%
^{-\theta z}\mathbb{W}_{q}\left( z\right)\mathrm{d} z\right).
\end{equation*}
We also have 
\begin{equation*}
\mathbb{Z}_{q}(x)=\mathbb{Z}_{\delta,q }\left( x,0\right) = 1 + q \int_0^x 
\mathbb{W}_q (y)\mathrm{d }y.
\end{equation*}
We denote the \textit{delayed $q$-scale function of $Y$} by 
\begin{equation*}
\Lambda_{\delta} ^{\left( q\right) }\left( x,s\right) =\int_{0}^{\infty }%
\mathbb{W}_{q}\left( x+z\right) \frac{z}{s}\mathbb{P}\left( X_{s}\in \mathrm{%
d}z\right).
\end{equation*}
In \cite{kyprianouloeffen2010} and \cite{renaud2014}, many fluctuation
identities for the refracted process are expressed in terms of the \textit{%
scale function} of $U$, that is, for $q \geq 0$ and for $x,z\in \mathbb{R}$,
set 
\begin{equation}  \label{small w}
w^{\left(q\right)}(x;z) = W_{q}(x-z) + \delta\mathbf{1}_{\left\lbrace x\geq
0\right\rbrace}\int^{x}_{0}\mathbb{W}_{q}(x-y)W_q^{^{\prime }}(y-z)\mathrm{d}
y .
\end{equation}
Note that when $x<0$, we have 
\begin{equation*}
w^{\left(q\right)}(x;z)=W_{q}(x-z) , 
\end{equation*}
and when $q=0$, we will write $w^{(0)}(x;z)=w(x;z)$. 
First, for $a \in \mathbb{R}$, we define the following first-passage
stopping times: 
\begin{align*}
\nu_a^- &= \inf\{t>0 \colon Y_t<a\} \quad \text{and} \quad \nu_a^+ =
\inf\{t>0 \colon Y_t\geq a\} \\
\kappa_a^- &= \inf\{t>0 \colon U_t<a\} \quad \text{and} \quad \kappa_a^+ =
\inf\{t>0 \colon U_t\geq a\}.
\end{align*}
Since $Y$ is also a spectrally negative L\'{e}vy process, the identities
for $X$ also hold for $Y$. For example, for $x\in \mathbb{R}$, 
\begin{equation}  \label{E:Paris001}
\mathbb{E}_{x}\left[ \mathrm{e}^{-\lambda \nu _{0}^{-}+rY_{\nu _{0}^{-}}}%
\right] =\mathbb{Z}_{\delta ,\lambda }\left( x,r\right) -\left( \frac{%
\psi_\lambda(r)-\delta r}{r-\varphi _{\lambda }}\right) \mathbb{W}_{\lambda
}\left( x\right).
\end{equation}
We denote by $\kappa_U^{q}$ the time of Parisian ruin with exponential
delays for the refracted L\'{e}vy process $U$ 
\begin{equation*}
\kappa_U^{q}=\inf \left\{ t>0 : t-g^{U}_{t}>\mathrm{e}^{g^{U}_{t}}_{q}
\right\}.
\end{equation*}
We have the following new results for the Laplace transforms of $\kappa_U^{q}
$ and $\mathcal{O}^{U}_{\mathrm{e}_{\lambda}}$.

\begin{lemma}
\label{lemmU} For $q,\lambda >0$ and $x\in \mathbb{R}$,  
\begin{equation*}
\mathbb{E}_{x}\left[ \mathrm{e}^{-\lambda \kappa _{U}^{q}}\right] =\frac{q}{
\lambda +q}\left( \mathbb{Z}_{\lambda }\left( x\right) -\dfrac{\lambda
\left( \Phi _{q+\lambda }-\varphi _{\lambda }\right) }{\left( q-\delta \Phi
_{\lambda +q}\right) \varphi _{\lambda }}\mathbb{Z}_{\delta ,\lambda }\left(
x,\Phi _{\lambda +q}\right) \right) ,
\end{equation*}
and consequently,  
\begin{equation}  \label{mainLapU}
\mathbb{E}_{x}\left[ \mathrm{e}^{-q\mathcal{O}^{U}_{\mathrm{e}_{\lambda }}}%
\right] =\dfrac{ q\lambda \left( \Phi _{q+\lambda }-\varphi _{\lambda
}\right) }{\left( \lambda +q\right) (q-\delta \Phi _{q+\lambda })\varphi
_{\lambda }}\mathbb{Z} _{\delta ,\lambda }\left( x,\Phi _{\lambda +q}\right)
-\dfrac{q\mathbb{Z} _{\lambda }(x)}{q+\lambda }+1.
\end{equation}
\end{lemma}

\begin{proof}
For $x<0$, using the strong Markov property of $U$ and the fact that $
U_{\kappa _{0}^{+}}=0$ on $\left\{ \kappa _{0}^{+}<\infty \right\} $, we 
have  
\begin{equation*}
\mathbb{E}_{x}\left[ \mathrm{e}^{-\lambda \kappa _{U}^{q}}\right] =\mathbb{E}
_{x}\left[ \mathrm{e}^{-\lambda \mathrm{e}_{q}}\mathbf{1}_{\left\{ \kappa
_{0}^{+}>\mathrm{e}_{q}\right\} }\right] +\mathbb{E}_{x}\left[ \mathrm{e}
^{-(q+\lambda )\kappa _{0}^{+}}\right] \mathbb{E}\left[ \mathrm{e}^{-\lambda
\kappa _{U}^{q}}\right] .
\end{equation*}
Since $\left\{ X_{t},t<\tau _{0}^{+}\right\} $ and $\left\{ U_{t},t<\kappa 
_{0}^{+}\right\} $ have the same distribution with respect to $\mathbb{P}%
_{x}  $ when $x<0$, we further have  
\begin{equation*}
\mathbb{E}_{x}\left[ \mathrm{e}^{-\lambda \kappa _{U}^{q}}\right] =\mathbb{E}
_{x}\left[ \mathrm{e}^{-\lambda \mathrm{e}_{q}}\mathbf{1}_{\left\{ \tau
_{0}^{+}>\mathrm{e}_{q}\right\} }\right] +\mathbb{E}_{x}\left[ \mathrm{e}
^{-(q+\lambda )\tau _{0}^{+}}\right] \mathbb{E}\left[ \mathrm{e}^{-\lambda
\kappa _{U}^{q}}\right] .
\end{equation*}
For $x\geq 0$, using the strong Markov property of $U$, we get  
\begin{eqnarray}
\mathbb{E}_{x}\left[ \mathrm{e}^{-\lambda \kappa _{U}^{q}}\right] &=& 
\mathbb{E}_{x}\left[ \mathrm{e}^{-\lambda \kappa _{0}^{-}}\mathbb{E}
_{U_{\kappa _{0}^{-}}}\left[ \mathrm{e}^{-\lambda \mathrm{e}_{q}}\mathbf{1}
_{\left\{ \tau _{0}^{+}>\mathrm{e}_{q}\right\} }\right] \right]  \notag \\
&&+\mathbb{E}_{x}\left[ \mathrm{e}^{-\lambda \kappa _{0}^{-}}\mathbb{E}
_{U_{\kappa _{0}^{-}}}\left[ \mathrm{e}^{-(q+\lambda )\tau _{0}^{+}}\right] %
\right] \mathbb{E}\left[ \mathrm{e}^{-\lambda \kappa _{U}^{q}}\right]  \notag
\\
&=&\frac{q}{q+\lambda }\left( \mathbb{E}_{x}\left[ \mathrm{e}^{-\lambda
\kappa _{0}^{-}}\right] -\mathbb{E}_{x}\left[ \mathrm{e}^{-\lambda \kappa
_{0}^{-}+\Phi _{\lambda +q}U_{\kappa _{0}^{-}}}\right] \right)  \notag \\
&&-\mathbb{E}_{x}\left[ \mathrm{e}^{-\lambda \kappa _{0}^{-}+\Phi _{\lambda
+q}U_{\kappa _{0}^{-}}}\right] \mathbb{E}\left[ \mathrm{e}^{-\lambda \kappa
_{U}^{q}}\right]  \notag \\
&=&\frac{q}{q+\lambda }\left( \mathbb{E}_{x}\left[ \mathrm{e}^{-\lambda \nu
_{0}^{-}}\right] -\mathbb{E}_{x}\left[ \mathrm{e}^{-\lambda \nu
_{0}^{-}+\Phi _{\lambda +q}Y_{\nu _{0}^{-}}}\right] \right)  \notag \\
&&-\mathbb{E}_{x}\left[ \mathrm{e}^{-\lambda \nu _{0}^{-}+\Phi _{\lambda
+q}Y_{\nu _{0}^{-}}}\right] \mathbb{E}\left[ \mathrm{e}^{-\lambda \kappa
_{U}^{q}}\right] ,  \label{alx}
\end{eqnarray}
where in the last equality we used the fact that $\left\{ Y_{t},t<\nu 
_{0}^{-}\right\} $ and $\left\{ U_{t},t<\kappa _{0}^{-}\right\} $ have the 
same distribution under $\mathbb{P}_{x}$. Note that the above expression 
holds for all $x\in \mathbb{R}$.

Now, we assume $X$ and $Y$ have paths of bounded variation. Solving for $ 
\mathbb{E}\left[ \mathrm{e}^{-\lambda \kappa _{U}^{q}}\right] $ and using  %
\eqref{E:Paris001}, we get  
\begin{eqnarray}
\mathbb{E}\left[ \mathrm{e}^{-\lambda \kappa _{U}^{q}}\right] &=&\frac{ 
\frac{q}{q+\lambda }\left( \mathbb{E}\left[ \mathrm{e}^{-\lambda \nu
_{0}^{-}}\right] -\mathbb{E}\left[ \mathrm{e}^{-\lambda \nu _{0}^{-}+\Phi
_{\lambda +q}Y_{\nu _{0}^{-}}}\right] \right) }{1-\mathbb{E}\left[ \mathrm{e}
^{-\lambda \nu _{0}^{-}+\Phi _{\lambda +q}Y_{\nu _{0}^{-}}}\right] }  \notag
\label{x0} \\
&=&\dfrac{q}{q+\lambda }-\dfrac{q}{\left( \lambda +q\right) }\frac{\lambda
\left( \Phi _{q+\lambda }-\varphi _{\lambda }\right) }{\varphi _{\lambda
}(q-\delta \Phi _{q+\lambda })}.
\end{eqnarray}
Substituting \eqref{E:Paris001} and \eqref{x0} into \eqref{alx}, we have  
\begin{eqnarray*}
\mathbb{E}_{x}\left[ \mathrm{e}^{-\lambda \kappa _{U}^{q}}\right] &=&\frac{q 
}{q+\lambda }\left( \mathbb{Z}_{\lambda }\left( x\right) -\frac{\lambda }{
\varphi _{\lambda }}\mathbb{W}_{\lambda }\left( x\right) \right) \\
&&-\frac{q}{q+\lambda }\left( \mathbb{Z}_{\delta ,\lambda }\left( x,\Phi
_{\lambda +q}\right) -\frac{(q-\delta \Phi _{q+\lambda })}{\left( \Phi
_{\lambda +q}-\varphi _{\lambda }\right) }\mathbb{W}_{\lambda }\left(
x\right) \right) \\
&&+\frac{q}{q+\lambda }\left( \mathbb{Z}_{\delta ,\lambda }\left( x,\Phi
_{\lambda +q}\right) -\frac{(q-\delta \Phi _{q+\lambda })}{\left( \Phi
_{\lambda +q}-\varphi _{\lambda }\right) }\mathbb{W}_{\lambda }\left(
x\right) \right) \\
&&-\dfrac{q}{\left( \lambda +q\right) }\frac{\lambda \left( \Phi _{q+\lambda
}-\varphi _{\lambda }\right) }{\varphi _{\lambda }(q-\delta \Phi _{q+\lambda
})}\left( \mathbb{Z}_{\delta ,\lambda }\left( x,\Phi _{\lambda +q}\right) - 
\frac{(q-\delta \Phi _{q+\lambda })}{\left( \Phi _{\lambda +q}-\varphi
_{\lambda }\right) }\mathbb{W}_{\lambda }\left( x\right) \right) \\
&=&\frac{q}{q+\lambda }\left( \mathbb{Z}_{\lambda }\left( x\right) -\frac{
\lambda }{\varphi _{\lambda }}\mathbb{W}_{\lambda }\left( x\right) \right) \\
&&-\dfrac{q}{\left( \lambda +q\right) }\frac{\lambda \left( \Phi _{q+\lambda
}-\varphi _{\lambda }\right) }{\varphi _{\lambda }(q-\delta \Phi _{q+\lambda
})}\left( \mathbb{Z}_{\delta ,\lambda }\left( x,\Phi _{\lambda +q}\right) - 
\frac{(q-\delta \Phi _{q+\lambda })}{\left( \Phi _{\lambda +q}-\varphi
_{\lambda }\right) }\mathbb{W}_{\lambda }\left( x\right) \right) \\
&=&\dfrac{q\mathbb{Z}_{\lambda }(x)}{q+\lambda }-\dfrac{q\lambda \left( \Phi
_{q+\lambda }-\varphi _{\lambda }\right) }{\left( \lambda +q\right) \varphi
_{\lambda }}\frac{\mathbb{Z}_{\delta ,\lambda }\left( x,\Phi _{\lambda
+q}\right) }{(q-\delta \Phi _{q+\lambda })}.
\end{eqnarray*}
The case where $X$ has paths of unbounded variation follows using the same 
approximating procedure as in \cite{kyprianouloeffen2010} (see also \cite%
{guerin_renaud_2015}).\newline
Finally, Equation \eqref{mainLapU} is immediate using again the following 
identity from Proposition $3.4$ in \cite{guerinrenaud2015}, namely
\begin{equation*}
\mathbb{E}_{x}\left[ \mathrm{e}^{-q\mathcal{O}^{U}_{\mathrm{e}_{\lambda }}}%
\right] =1-\mathbb{\ E}_{x}\left[ \mathrm{e}^{-\lambda \kappa _{U}^{q}}%
\right] .
\end{equation*}
\end{proof}

Using similar techniques as in the proof of Theorem \ref{maintheo}, we
obtain the following expression for the distribution of $\mathcal{O}^{U}_{%
\mathrm{e} _{\lambda }}$. The result is stated without proof. We point out
that Equations \eqref{mainLapU} and \eqref{mainresU} generalize Corollary 2
of Kyprianou et al. \cite{kyprianou2014occupation} in which the occupation
time is up to an infinite time horizon.

\begin{theorem}
\label{maintheoU} For $\lambda >0$, $x \in \mathbb{R} $ and $y \geq 0$,  
\begin{eqnarray}  \label{mainresU}
\mathbb{P}_{x}\left( \mathcal{O}^{U}_{\mathrm{e}_{\lambda }}\in \mathrm{d}%
y\right) &=& \left( \mathbb{Z}_{\lambda }(x)-\frac{\lambda }{\varphi_\lambda 
} \mathbb{Z}_{\lambda }\left( x\right) \right)\delta _{0}\left( \mathrm{d}%
y\right)  \notag \\
&&+\lambda \mathrm{e}^{-\lambda y}\left( \mathcal{B}_\delta ^{\left( \lambda
\right) }\left( x,y\right) -\lambda \int_{0}^{y}\mathcal{B}_\delta ^{\left(
\lambda \right) }\left( x,s\right) \mathrm{d}s \right) \mathrm{d} y  \notag
\\
&&+\lambda \mathrm{e}^{-\lambda y}\left(\mathbb{Z}_{\lambda}(x)-\frac{%
\lambda }{ \varphi_\lambda }\mathbb{Z}_{\lambda }\left( x\right) \right)%
\mathrm{d} y,
\end{eqnarray}
where  
\begin{eqnarray*}
\mathcal{B}^{\left(\lambda\right) }_{\delta}\left(x,s\right) =\dfrac{
\Lambda_{\delta}^{\left( \lambda \right) \prime }\left( x,s\right)}{
\varphi_\lambda} -\Lambda_{\delta}^{\left( \lambda \right)}\left( x,s\right).
\end{eqnarray*}
\end{theorem}

We denote the inverse occupation time of the refracted process $U$ by 
\begin{equation*}
\sigma^{U}_{r}=\inf \left\{ t>0 \colon \mathcal{O}^{U}_t > r \right\},
\end{equation*}
and for which we obtain the following Laplace transform.

\begin{theorem}
\label{maintheoParisianU} For $r,\lambda >0$ and $x \in \mathbb{R} $,  
\begin{eqnarray}  \label{LapcumuU}
\mathbb{E}_x\left[ \mathrm{e}^{-\lambda \sigma ^{U}_{r}}\right] &=& \mathrm{e}^{-\lambda r}\left( \mathbb{Z}_{\lambda }(x)-\frac{\lambda }{\varphi_\lambda 
} \mathbb{W }_{\lambda }\left( x\right) \right)  \notag \\
&&-\lambda \int_{0}^{r}\mathrm{e}^{-\lambda u}\left( \mathcal{B}
_\delta^{\left( \lambda \right) }\left( x,u\right) -\lambda \int_{0}^{u} 
\mathcal{B}_\delta^{\left( \lambda \right) }\left( x,s\right) \mathrm{d}
s\right) \mathrm{d}u.
\end{eqnarray}
\end{theorem}

We also obtain the following expression of the probability of inverse
occupation time for the refracted process $U$.

\begin{corol}
\label{LaplaceTrU}  For $r>0$, $x \in \mathbb{R} $ and $\mathbb{E}%
[X_1]>\delta$,  
\begin{eqnarray}  \label{cumuParruinU}
\mathbb{P}_x\left( \sigma^{U} _{r} < \infty \right) =1-\left(\mathbb{E}%
[X_1]-\delta\right)\left( \mathbb{W}(x)+\int_{0}^{r} \Lambda_{\delta}^{
\prime }\left( x,s\right)\mathrm{d} s\right).
\end{eqnarray}
\end{corol}

It is a trivial exercise to show that when $\delta=0$, the results reduced
to those given in Section \ref{section1}.

\begin{remark}
The above expression can also be expressed as follows,  
\begin{equation}  \label{cummU}
\mathbb{P}_{x}\left( \sigma _{r}^{U}<\infty \right) =1-\left( \mathbb{E}%
[X_1]-\delta \right) \left( \dfrac{w(x;0)}{1-\delta W(0)}+\int_{0}^{r}%
\Lambda _{\delta }^{\prime }\left( x,s\right) \mathrm{d} s\right) ,
\end{equation}
which is due to the following useful identity relating different scale 
functions and taken from \cite{renaud2014}, that is, for $p,q\geq 0$ and $
x\in \mathbb{R}$,  
\begin{equation*}
(q-p)\int_{0}^{x}\mathbb{W}_{p}(x-y)W_{q}(y)\mathrm{d}y=W_{q}(x)-\mathbb{W}
_{p}(x)+\delta \left( W_{q}(0)\mathbb{W}_{p}(x)+\int_{0}^{x}\mathbb{W}
_{p}(x-y)W_{q}^{\prime }(y)\mathrm{d}y\right) ,
\end{equation*}
and for which we consider $p=q=0$. Note that when $\delta =0$, we recover 
the spectrally negative analogue in \eqref{convsln}. Letting $r\rightarrow 0$
in \eqref{cummU}, we recover the classical probability of ruin of $U$  
\begin{equation*}
\mathbb{P}_{x}\left( \kappa _{0}^{-}<\infty \right) =1-\frac{\left( \mathbb{E%
}\left[ X_{1} \right] -\delta \right) }{1-\delta W(0)}w(x;0).
\end{equation*}
\end{remark}
\subsection{Examples\label{sectionexamprefracted}}

As mentioned in the introduction, to the best of our knowledge, the distribution of the finite horizon occupation time is known only for Brownian motions with drift and Cram\'{e}r-Lundberg process with exponential claims in the literature. In this
section, by a further inversion of \eqref{mainresU}, we are able to derive
the distribution of the finite horizon occupation time formula for refracted
Brownian risk process and refracted Cram\'{e}r-Lundberg process with
exponential claims. Both formulas are new in the literature.

\subsubsection{A refracted Brownian risk process}

Let $X$ and $Y$ be two Brownian risk processes, defined as
\begin{equation*}
X_t - X_0 = \mu t +\sigma B_t \quad \text{and} \quad Y_t -Y_0 = (\mu-\delta)
t + \sigma B_t , 
\end{equation*}
where $B=\{B_t, t\geq 0\}$ is a standard Brownian motion. In this case, the
Laplace exponent of $X$ is given by 
\begin{equation*}
\psi(\lambda) = \mu \lambda + \frac12\sigma^2 \lambda^2. 
\end{equation*}
Then, for $x\geq 0$, we have 
\begin{align*}
W(x) &= \frac{1}{\mu} \left( 1-\mathrm{e}^{-2\frac{\mu}{\sigma^2} x} \right),
\\
\mathbb{W}(x) &= \frac{1}{\mu-\delta} \left( 1-\mathrm{e}^{-2\frac{\mu-\delta%
}{\sigma^2} x} \right).
\end{align*}
and $\varphi _{\lambda }=\sigma ^{-2}\left( \sqrt{\left( \mu -\delta \right)
^{2}+2\lambda \sigma ^{2}}-\left( \mu -\delta \right) \right).$ Using the
fact that 
\begin{equation*}
\int_{0}^{\infty }\frac{z}{r}\mathbb{P}(X_{r}\in \mathrm{d}z)=\frac{\sigma }{%
\sqrt{2\pi r}}\mathrm{e}^{-\frac{c^{2}r}{2\sigma ^{2}}}+c\mathcal{N}\left( 
\frac{c\sqrt{r}}{\sigma }\right),
\end{equation*}
and 
\begin{equation*}
\int_{0}^{\infty }\frac{z}{r}\mathrm{e}^{-\frac{2\left( \mu -\delta \right) 
}{\sigma ^{2}}z}\mathbb{P}(X_{r}\in \mathrm{d}z)=\frac{\sigma }{\sqrt{2\pi r}%
}\mathrm{e}^{-\frac{c^{2}r}{2\sigma ^{2}}}+\frac{\left( 2\delta -\mu \right) 
}{\sqrt{2\pi }}\mathcal{N}\left( \sqrt{r}\frac{\left( 2\delta -\mu \right) }{%
\sigma }\right),
\end{equation*}
we obtain 
\begin{eqnarray*}
\Lambda _{\delta }\left( x,s\right) &=&\int_{0}^{\infty }\mathbb{W}\left(
x+z\right) \frac{z}{r}\mathbb{P}\left( X_{r}\in \mathrm{d}z\right) \\
&=&\frac{1}{\mu -\delta }\left( \frac{\sigma }{\sqrt{2\pi r}}\mathrm{e}^{-%
\frac{\mu ^{2}r}{2\sigma ^{2}}}+\mu \mathcal{N}\left( \frac{\mu \sqrt{r}}{%
\sigma }\right) \right) \\
&&-\frac{\mathrm{e}^{-\frac{2\left( \mu -\delta \right) }{\sigma ^{2}}x}}{%
\mu -\delta }\left( \frac{\sigma }{\sqrt{2\pi r}}\mathrm{e}^{-\frac{\mu ^{2}r%
}{2\sigma ^{2}}}+\mathrm{e}^{\frac{2r\delta \left( \delta -2\mu \right) }{%
\sigma ^{2}}}\frac{\left( 2\delta -\mu \right) }{\sqrt{2\pi }}\mathcal{N}%
\left( \sqrt{r}\frac{\left( 2\delta -\mu \right) }{\sigma }\right) \right)
\end{eqnarray*}%
Hence, 
\begin{equation*}
\Lambda _{\delta }^{\prime }\left( x,s\right) =\frac{2}{\sigma ^{2}}\mathrm{e%
}^{-\frac{2\left( \mu -\delta \right) }{\sigma ^{2}}x}\left( \frac{\sigma }{%
\sqrt{2\pi r}}\mathrm{e}^{-\frac{\mu ^{2}r}{2\sigma ^{2}}}+\mathrm{e}^{\frac{%
2r\delta \left( \delta -2\mu \right) }{\sigma ^{2}}}\frac{\left( 2\delta
-\mu \right) }{\sqrt{2\pi }}\mathcal{N}\left( \sqrt{r}\frac{\left( 2\delta
-\mu \right) }{\sigma }\right) \right) .
\end{equation*}
We also have 
\begin{equation*}
\frac{\mathrm{e}^{-\lambda r}}{\varphi _{\lambda }}=\int_{0}^{\infty }%
\mathrm{e}^{-\lambda t}\left( \left( \mu -\delta \right) +\frac{\sigma 
\mathrm{e}^{-\left( \mu -\delta \right) ^{2}/2\sigma ^{2}\left( t-s\right) }%
}{\sqrt{2\pi \left( t-s\right) }}-\left( \mu -\delta \right) \mathcal{N}%
\left( \frac{\left( \mu -\delta \right) \sqrt{r}}{\sigma }\right) \right) 
\mathrm{d}t
\end{equation*}
Using Laplace inversion, we finally obtain 
\begin{eqnarray*}
\mathbb{P}\left( \mathcal{O}^{U}_{t}\in \mathrm{d} s\right) &=&\frac{2}{%
\sigma ^{2}}\left( \frac{\sigma }{\sqrt{2\pi r}}\mathrm{e}^{-\frac{\mu ^{2}r%
}{2\sigma ^{2}}}+\mathrm{e}^{\frac{2r\delta \left( \delta -2\mu \right) }{%
\sigma ^{2}}}\frac{\left( 2\delta -\mu \right) }{\sqrt{2\pi }}\mathcal{N}%
\left( \sqrt{r}\frac{\left( 2\delta -\mu \right) }{\sigma }\right) \right) \\
&&\times \left( \left( \mu -\delta \right) +\frac{\sigma \mathrm{e}^{-\left(
\mu -\delta \right) ^{2}/2\sigma ^{2}\left( t-s\right) }}{\sqrt{2\pi \left(
t-s\right) }}-\left( \mu -\delta \right) \mathcal{N}\left( \frac{\left( \mu
-\delta \right) \sqrt{r}}{\sigma }\right) \right) \mathrm{d}s,
\end{eqnarray*}

\subsubsection{A refracted Cram\'{e}r-Lundberg process with exponential claims}

Let $X$ and $Y$ be two Cram\'{e}r-Lundberg risk processes with exponentially distributed claims, then they are defined as 
\begin{equation*}
X_t - X_0 = c t - \sum_{i=1}^{N_t} C_i \quad \text{and} \quad Y_t - Y_0 =
(c-\delta) t - \sum_{i=1}^{N_t} C_i , 
\end{equation*}
where $N=\{N_t, t\geq 0\}$ is a Poisson process with intensity $\eta>0$, and $\{C_1, C_2, \dots\}$ are independent and exponentially distributed
random variables with mean $1/\alpha$, independent of $N$. In this case, the Laplace exponent of $X$ is given by 
\begin{equation*}
\psi(\lambda) = c \lambda + \eta \left( \frac{\alpha}{\lambda+\alpha} - 1
\right) , \quad \text{for $\lambda>-\alpha$} 
\end{equation*}
and the net profit condition is given by $\mathbb{E} \left[ Y_1 \right]
=c_\delta- \eta/\alpha \geq 0$ where $c_\delta= c - \delta$. Then, for $%
x\geq 0$, we have 
\begin{align*}
W(x) &= \frac{1}{c-\eta/\alpha} \left( 1- \frac{\eta}{c\alpha}\mathrm{e}^{(%
\frac{\eta}{c}-\alpha)x} \right), \\
\mathbb{W}(x) &= \frac{1}{c_\delta-\eta/\alpha} \left( 1- \frac{\eta}{%
c_\delta\alpha}\mathrm{e}^{\left(\frac{\eta}{c_\delta}-\alpha\right)x}
\right), \\
\end{align*}
The right-inverse is given by 
\begin{equation*}
\varphi _{\lambda }=\frac{1}{2c_\delta}\left( \lambda +\eta -c_\delta\alpha +%
\sqrt{\left(\lambda +\eta -c_\delta\alpha \right) ^{2}+4c_\delta)\alpha
\lambda }\right) .
\end{equation*}%
We have 
\begin{eqnarray*}
\Lambda _{\delta }\left( x,s\right) &=&\int_{0}^{\infty }\mathbb{W}\left(
x+z\right) \frac{z}{r}\mathbb{P}\left( X_{r}\in \mathrm{d}z\right) \\
&=&\frac{1}{c_\delta -\eta /\alpha }\int_{0}^{\infty }\left( 1-\frac{\eta }{%
\left( c_\delta \right) \alpha }\mathrm{e}^{\left( \frac{\eta }{ c_\delta }%
-\alpha \right) \left( x+z\right) }\right) \frac{z}{r}\mathbb{P}\left(
X_{r}\in \mathrm{d}z\right) ,
\end{eqnarray*}%
where 
\begin{eqnarray*}
&&\int_{0}^{\infty }z\mathrm{e}^{\left( \frac{\eta }{c_{\delta }}-\alpha
\right) z}\mathbb{P}\left( X_{r}\in \mathrm{d}z\right) \\
&=&\int_{0}^{cs}z\mathrm{e}^{\left( \frac{\eta }{c_{\delta }}-\alpha \right)
z}\mathrm{e}^{-\eta r}\left( \delta _{0}\left( cr-\mathrm{d}z\right) +%
\mathrm{e}^{-\alpha (cr-z)}\sum_{m=0}^{\infty }\frac{(\alpha \eta r)^{m+1}}{%
m!(m+1)!}(cr-z)^{m}\mathrm{d}z\right) \\
&=&\int_{0}^{cs}\mathrm{e}^{-\eta r}\left( cr\mathrm{e}^{\left( \frac{\eta }{%
c_{\delta }}-\alpha \right) cr}+\mathrm{e}^{-\alpha cr}z\mathrm{e}^{\frac{%
\eta }{c_{\delta }}z}\sum_{m=0}^{\infty }\frac{(\alpha \eta r)^{m+1}}{%
m!(m+1)!}(cr-z)^{m}\mathrm{d}z\right) \\
&=&\mathrm{e}^{-r\left( \eta +\alpha c+\frac{\eta c}{c_{\delta }}\right)
}\left( cr+\sum_{m=0}^{\infty }\frac{(\alpha \eta r)^{m+1}}{m!(m+1)!}%
\int_{0}^{cr}\mathrm{e}^{-\frac{\eta }{c_{\delta }}y}y^{m}\left( cr-y\right) 
\mathrm{d}z\right) \\
&=&cr\mathrm{e}^{-\eta r+\left( \frac{\eta }{c_{\delta }}-\alpha \right) cr}
\\
&&\times \left( cr+\sum_{m=0}^{\infty }\frac{(\eta s)^{m+1}}{m!(m+1)!}\left[
cr\Gamma (m+1,\frac{\eta cr}{c_{\delta }})-\frac{c_{\delta }}{\left( \eta
+c_{\delta }\alpha cr\right) }\Gamma (m+2,\frac{cr\eta }{c_{\delta }})\right]
\right) .
\end{eqnarray*}
Then, 
\begin{multline*}
\Lambda ^{\prime }\left( 0,s\right) =-\frac{\eta c\mathrm{e}^{-\eta
r+\left( \frac{\eta }{c_{\delta }}-\alpha \right) cr}}{c_{\delta }^{2}}+%
\frac{\eta \mathrm{e}^{-\eta r+\left( \frac{\eta }{c_{\delta }}-\alpha
\right) cr}}{c_{\delta }^{2}} \\
\times \sum_{m=0}^{\infty }\frac{(\eta s)^{m+1}}{m!(m+1)!}\left[ cr\Gamma
(m+1,\frac{\eta cr}{c_{\delta }})-\frac{c_{\delta }}{\left( \eta +c_{\delta
}\alpha cr\right) }\Gamma (m+2,\frac{cr\eta }{c_{\delta }})\right] .
\end{multline*}
Since $X$ is of bounded variation paths (i.e., $\mathbb{W}_{\lambda }(0)>0$%
), we have 
\begin{eqnarray*}
\lambda ^{-1}\mathbb{P}\left( \mathcal{O}^{U}_{\mathrm{e}_{\lambda }}\in 
\mathrm{d}s\right)=\frac{1}{\varphi _{\lambda }}\mathbb{W}_{\lambda }\left(
0\right) \delta _{0}\left( \mathrm{d}s\right) +\frac{\mathrm{e}^{-\lambda s}%
}{\varphi _{\lambda }}\Lambda _{\delta }^{\prime }\left( 0,s\right) \mathrm{d%
}s.
\end{eqnarray*}%
Also, we have 
\begin{equation*}
\frac{1}{\varphi _{\lambda } c_\delta }=\int_{0}^{\infty }\mathrm{e}%
^{-\lambda t}a_{t}^{\delta }\mathrm{d}t,
\end{equation*}%
where 
\begin{equation*}
a_{t}^{\delta }=\left( 1-\frac{\eta }{c_\delta \alpha }\right) _{+}+\frac{2
\eta }{\pi }\mathrm{e}^{-\left( \eta +c_\delta \alpha \right) t}\int_{-1}^{1}%
\frac{\sqrt{1-u^{2}}\mathrm{e}^{-2\sqrt{ c_\delta \alpha \eta }tu}}{\eta
+c_\delta \alpha +2\sqrt{c_\delta \alpha \eta }u}\mathrm{d}t.
\end{equation*}%
We then obtain the following expression for the distribution of the
occupation time $\mathcal{O}^{U}_{t}$ for $t>0$ :
\begin{equation*}
\mathbb{P}\left( \mathcal{O}^{U}_{t}\in \mathrm{d}s\right)
=a^{\delta}_{t}\delta _{0}\left( \mathrm{d} s\right) +c_\delta\Lambda_\delta
^{\prime }\left( 0,s\right) a^{\delta}_{t-s}\mathbf{1}_{\left( 0,t\right)
}\left( s\right) \mathrm{d} s.
\end{equation*}
\section*{Acknowledgements}

Support from grants from the Natural Sciences and Engineering Research
Council of Canada is gratefully acknowledged by David Landriault and Bin Li
(grant numbers 341316 and 05828, respectively). Support from the Canada
Research Chair Program is gratefully acknowledged by David Landriault.

\bibliographystyle{alpha}
\bibliography{REFERENCES}

\end{document}